\documentclass[11pt]{article}
\usepackage{amsmath}
\usepackage{amsthm}
\usepackage{amsfonts}
\usepackage{xspace}
\usepackage{algorithm}
\usepackage[noend]{algorithmic}
\usepackage{geometry}
\usepackage{authblk}
\usepackage{hyperref}
\usepackage{graphicx}
\newtheorem{theorem}{Theorem}
\newtheorem{lemma}{Lemma}
\newtheorem{definition}{Definition}
\newtheorem{claim}{Claim}

\newcommand{\eps}{\epsilon}
\newcommand{\cgp}{\textsc{Capacitated Graph Pricing}\xspace}
\newcommand{\lsp}{\textsc{L-Sided Pricing}\xspace}

\begin{document}
\title{Graph Pricing With Limited Supply}
%
%
\author{Zachary Friggstad\footnote{This research was undertaken, in part, thanks to funding from the Canada Research Chairs program and an NSERC Discovery Grant.} }
\author{Maryam Mahboub}
\affil{Department of Computing Science, University of Alberta}
\date{}
\maketitle              
\begin{abstract}
We study approximation algorithms for graph pricing with vertex capacities yet without the traditional envy-free constraint.
Specifically, we have a set of items $V$ and a set of customers $X$ where each customer $i \in X$ has a budget $b_i$
and is interested in a bundle of items $S_i \subseteq V$ with $|S_i| \leq 2$. However, there is a limited supply of each item: we only have $\mu_v$ copies of item $v$ to sell for each $v \in V$.
We should assign prices $p(v)$ to each $v \in V$ and chose a subset $Y \subseteq X$ of customers so that each $i \in Y$ can afford their bundle ($p(S_i) \leq b_i$) and at most $\mu_v$ chosen customers have item $v$ in their bundle for
each item $v \in V$. Each customer $i \in Y$ pays $p(S_i)$ for the bundle they purchased: our goal is to do this in a way that maximizes revenue.
Such pricing problems have been studied from the perspective of {\em envy-freeness} where we also must ensure that $p(S_i) \geq b_i$ for each $i \notin Y$. However, the version where we simply allocate items to customers after setting prices and do not worry
about the envy-free condition has received less attention.

With unlimited supply of each $v \in V$, Balcan and Blum (2006) give a 4-approximation for graph pricing which was later shown to be tight by Lee (2015) unless the Unique Games conjecture fails to hold. Our main result is an 8-approximation
for the capacitated case via local search and a 7.8096-approximation in simple graphs with uniform vertex capacities.
The latter is obtained by combing a more involved analysis of a multi-swap local search algorithm for constant capacities and an LP-rounding algorithm for larger capacities.
If all capacities are bounded by a constant $C$, we further show
a multi-swap local search algorithm yields an $\left(4 \cdot \frac{2C-1}{C} + \epsilon\right)$-approximation.
We also give a $(4+\eps)$-approximation in simple graphs through LP rounding when all capacities are very large as a function of $\eps$.

We use a reduction by Balcan and Blum to the case of bipartite graphs where all items on one side must be assigned a price of 0 holds in this setting as well. However, unlike their setting, the resulting problem remains {\bf APX}-hard even if all items have at most 4 copies to sell. We also show our multi-swap analysis is tight using an interesting construction based on regular, high-girth graphs.
\end{abstract}
\section{Introduction}
Choosing prices to sell items in order to maximize revenue is a complicated task even in environments where one can be certain of customer behaviour. Indeed, many so-called {\em pricing problems} have been studied in combinatorial optimization. One popular setting is this:
a collection of items $V$ is available to be sold where we have $\mu_v \in \mathbb Z_{\geq 0} \cup \{\infty\}$ copies of item $v \in V$. Additionally, we are given a collection of customers $X$ where each $i \in X$ has some budget $b_i \geq 0$. In the {\em single-minded} setting,
each customer $i \in X$ is interested in a bundle $S_i \subseteq V$.
We must assign prices $p : V \rightarrow \mathbb R_{\geq 0}$
to the items and sell them to some customers $Y \subseteq X$ while respecting two constraints:
\begin{itemize}
\item {\bf Affordability}:~$p(S_i) := \sum_{v \in S_i} p(v) \leq b_i$ for $i \in Y$, and
\item {\bf Supply Constraints}:~$|\{ i \in Y : v \in S_i\}| \leq \mu_v$ for $v \in V$.
\end{itemize}
That is, each customer that is allocated their bundle can afford it and no item is oversold.
Such a solution $(p, Y)$ is said to be {\bf feasible}, and the goal
is to find a feasible $(p, Y)$ maximizing revenue, {\em i.e.} $\sum_{i \in Y} p(S_i)$.

Much attention has been given to the {\bf envy-free} setting, where a feasible solution must additionally satisfy the property
$p(S_i) \geq b_i$ for $i \notin Y$ or to the {\bf unlimited supply} setting where $\mu_v = \infty$ for each $v \in V$.
Observe that in the unlimited supply setting, any pricing yields an envy-free solution by simply choosing the customers that can afford the price.
However, the problem still remains hard in the setting with unlimited supply of each item, see the related works section.

The single-minded, envy-free pricing (SMEFP) problem with limited supply was studied by Cheung and Swamy \cite{CS08}.
Somewhat informally, they show the following. If there is an LP-based $\alpha$-approximation to the problem of choosing the best customers $Y$ when given prices $p$ (without regard to the envy-free condition), then there is an $O(\alpha \cdot \log C)$-approximation to SMEFP where $C = \max_{v \in V} \mu_v$. In the special case where $|S_i|$ is bounded by a constant for each $i \in X$, this yields a logarithmic approximation for SMEFP.

We study single-minded pricing problems yet {\em without} the envy-free constraint. This is a natural variant of pricing problems where customer satisfaction is less of a concern than overall revenue generation. 
To the best of our knowledge, it seems that pricing problems without the envy-free condition like this have received virtually no attention so far except in simpler cases of unlimited supply where envy-freeness
is a superfluous constraint ({\em i.e.} any solution can be trivially be made envy-free without losing revenue).

More specifically we mainly consider the case when $|S_i| = 2$ for each customer $i$. Here, the set of customers can be thought of as edges $E$ in a graph
$G = (V,E)$ with vertex capacities indicating the number of copies of the item/vertex that can be sold. We show that without the envy-free condition, the problem admits a constant-factor approximation.
In fact, it is relatively easy to get a 16-approximation using randomized rounding (with alterations) of an LP relaxation.
Our primary focus is on obtaining smaller constants through a more intricate {\em local search} approximation algorithm. We also show how a combined approach using linear programming and local search can yield even better approximations
in certain settings.


\subsection{Our Results}\label{sec:results}

We use shorthand notation like $e = uv \in E$ when we want to consider an edge $e \in E$ in some graph $G = (V,E)$ and also want name the endpoints $u,v$ of $e$.
This allows us to name distinct customers (i.e. $e$) who are interested in the same bundle of items ({\em i.e.} $\{u,v\}$).
As mentioned earlier, we focus on the following problem.
\begin{definition}
Let $G = (V,E)$ be a graph with vertex capacities $\mu_v \in \mathbb Z_{\geq 0} \cup \{\infty\}$ where each $e = uv \in E$ has a budget $b_e \geq 0$
and is interested in the bundle of vertices $\{u,v\}$.
In \cgp, we want to find a pricing $p : V \rightarrow \mathbb R_{\geq 0}$ and $F \subseteq E$ such that $(p,F)$ is a feasible
solution to the pricing problem. The goal is to maximize revenue: $\sum_{e = uv \in F} p(u) + p(v)$.
\end{definition}
All of our algorithmic results extend in a simple way to the case where each customer is interested in a bundle of size {\em at most} 2,
but it is slightly simpler to describe the algorithms and their analysis
for the case where each customer wants precisely two different items. Unless otherwise stated,
the graph $G$ may have parallel edges. We use the term {\em simple graph}
to indicate it does not have parallel edges, meaning no two customers want exactly the same bundle.

To get approximations for \cgp, we use the reduction from Balcan and Blum \cite{BB06} to reduce to the case of a bipartite graph where
all items on one side must be given a price of 0. Specifically, we consider the following problem.
\begin{definition}\label{def:lsp}
In \lsp, we are given a \cgp instance in a bipartite graph $(L \cup R, E)$. A feasible solution $(p,F)$ must also have $p(v) = 0$ for each $v \in R$.
\end{definition}

Though they only focused on uncapacitated pricing problems, the reduction in \cite{BB06} easily shows an $\alpha$-approximation for \lsp yields a $4 \alpha$-approximation for \cgp. We remark that the reduction would not be valid if one were looking for envy-free solutions. A key difference between the uncapacitated case studied in \cite{BB06} and the capacitated case we work with is that \lsp remains
hard if there are capacities.
\begin{theorem}\label{thm:apx}
\lsp is {\bf APX}-hard, even if all capacities are at most 4 and all customers have a budget of 1 or 2.
\end{theorem}

So we turn to approximation algorithms. It is possible to get a 4-approximation for \lsp through straightforward rounding of a natural linear programming relaxation that is presented in Section \ref{sec:lp}, thus leading
to a 16-approximation for \cgp overall. We consider an alternative approach to get better approximation guarantees.

Our main algorithmic results are summarized as follows.
\begin{theorem} \label{thm:single}
There is a polynomial-time 2-approximation for \lsp.
\end{theorem}
This 2-approximation is fairly simple to obtain using local search. But we think it nicely highlights a direction for to designing approximations for pricing+packing problems.

We obtain slightly improved bounds in simple graphs with uniform capacities.
\begin{theorem}\label{thm:uniform}
Instances of \lsp in simple graphs where all vertices in $L \cup R$ have the same capacity $\mu$ admit a randomized, polynomial-time $1.952381$-approximation. For \cgp, this yields a $7.8096$-approximation.
\end{theorem}
This is obtained through a hybrid of local-search techniques and LP rounding techniques: if $\mu$ is bounded by some appropriate constant then a multi-swap local search
algorithm is used, otherwise we use a randomized LP rounding approach.

For the first part, we consider a more involved algorithm for \lsp with bounded capacities (that need not be uniform).
\begin{theorem}\label{thm:multi}
For any constants $C \geq 2, \epsilon > 0$, there is a polynomial-time $\left(\frac{2C-1}{C} + \epsilon\right)$-approximation for instances of \lsp with $\mu_v \leq C$ for all $v \in L$.
\end{theorem}
Note this does not require any bounds on capacities for nodes in $R$.
For example with $C = 2$ this yields a $(1.5 + \epsilon)$-approximation and in the case we prove is {\bf APX}-hard ($C = 4$) this yields
a $(1.75 + \epsilon)$-approximation. Observe if $C=1$ then both \cgp and \lsp reduce to maximum-weight matching because we can easily set prices to match the full budget of all edges in any matching.

Theorems \ref{thm:single} and \ref{thm:multi} are proven using {\em local search} algorithms.
That is, if we are given prices $p : L \rightarrow \mathbb R_{\geq 0}$
then the optimal customers $F \subseteq E$ can be computed using a maximum-weight $\mu$-matching algorithm. The local search algorithm
for Theorem \ref{thm:single} iteratively tries to change the price of one item in $L$ to see if it yields a better matching. We prove
with a simple argument that a local optimum is a 2-approximate solution for \lsp.

To prove Theorem \ref{thm:multi}, we consider a local search algorithm that changes $O_{\epsilon, C}(1)$ prices at a time in each step.
To analyze the performance of such an algorithm, we need a result about covering directed graphs by directed balls in a uniform way.
This may be of independent interest in other settings, so we state it here.

Let $H = (L,F)$ be a directed graph. For any $u \in L$ and $r \geq 0$ consider the ``directed ball'' $B^+(u,r) = \{v \in L: d_H(u,v) \leq r\}$ of nodes in $L$ reachable from $u$ in at most $r$ steps.
Similarly, let $\partial B^+(u,r) = \{v \in L: d(u,v) = r\}$ be nodes $v$ such that the shortest $u-v$ path in $H$ has length exactly $r$ (the {\em boundary} of $B^+(u,r)$).
We prove the following covering result for directed graphs.

\begin{theorem}\label{thm:cover}
Let $H = (L,F)$ be a directed graph where the indegree of each node is at most $C$ and let $d \in \mathbb Z_{\geq 0}$.
There is a ``weighting'' of directed balls $\tau : L \times \{0, 1, \ldots, d\} \rightarrow \mathbb Z_{\geq 0}$ with the following properties:
\begin{itemize}
\item For any $v \in V$, $\displaystyle\sum_{\substack{u \in L, 0 \leq r \leq d\\ \texttt{s.t. } v \in B^+(u,r)}} \tau(u,r) = \sum_{i=0}^d C^i = \frac{C^{d+1}-1}{C-1}$.
\item For any $v \in V$, $\displaystyle\sum_{\substack{u \in L, 0 \leq r \leq d\\ \texttt{s.t. } v \in \partial B^+(u,r)}} \tau(u,r) = C^d$.
\end{itemize}
Furthermore, $\tau(u,r) \leq C^{d-r}$ for each $u \in V$ and $0 \leq r \leq d$.
\end{theorem}
That is, each $v \in L$ lies in these balls with weight precisely $\frac{C^{d+1}-1}{C-1}$ and appears on the boundary of the balls with weight precisely $C^d$.
The bound on $\tau(u,r)$ at the end of the statement is a required to ensure the local search algorithm used to prove Theorem \ref{thm:multi} runs in polynomial time.

We also show the analysis of both algorithms are tight. See Section \ref{sec:algo} for definitions of the two local search algorithms
mentioned in the results below and Section \ref{sec:locality} for precise statements of how the analysis is tight.
Slightly informal statements of the results are mentioned here.
\begin{theorem}
For any $\epsilon > 0$, the locality gap of the single-swap algorithm (Algorithm \ref{alg:single}) can be as bad as $2-\epsilon$.
\end{theorem}
\begin{theorem}\label{thm:lowermulti}
For any $C \geq 2, \rho \geq 1, \epsilon > 0$, the locality gap of the $\rho$-swap algorithm (Algorithm \ref{alg:multi}) on \lsp instances
having maximum capacity at most $C$ can be as bad as $\left(\frac{2C-1}{C} - \epsilon\right)$
\end{theorem}
The first construction is quite simple, but the second construction for the multi-swap analysis is much more involved. As a starting point
for the construction, we require simple, regular graphs of constant degree but arbitrarily large girth. Such graphs were shown to exist by Sachs \cite{S63}.

Instances of $\lsp$ satisfying $\mu_v = \infty$ for $v \in R$ can be solved in polynomial time by a greedy algorithm
that simply chooses the best price at each vertex: customers incident to different vertices in $L$ have no interaction \cite{BB06}.
So it is natural to wonder if we can get better approximations if capacities are large.
Though it is not our main result, we confirm the intuition that larger capacities yield better approximations (in simple graphs)
through randomized rounding of an LP relaxation.
\begin{theorem}\label{thm:rounding}
For any $\epsilon > 0$, there is some $C_{\epsilon} \geq 0$ such that instances of \cgp in simple graphs satisfying
$\mu_v \geq C_\epsilon$ admit a randomized, polynomial-time $(4+\epsilon)$-approximation.
\end{theorem}



\subsection{Related Work}
The basic model of pricing problems of this sort were introduced by Guruswami {\em et al.} \cite{GHKKKM05}.
Among other results, they give an $O(\log n + \log m)$-approximation for the case of single-minded pricing without item capacities if we have $n$ items and $m$ customers. Here, the bundle $S_i$ for each customer $i$ may be any subset of items (not just size 2). This was later improved by Briest and Krysta to an $O(\log D + \log k)$-approximation where each set has size at most $k$ and each item appears in at most $D$ sets \cite{BK06}.
The logarithmic approximation is essentially tight: Chalermsook {\em et al.} show for any constant $\epsilon > 0$ there is no $O(\log^{1-\epsilon}(m + n))$-approximation unless {\bf NP} $\subseteq$ {\bf DTIME}$(n^{\texttt{polylog}(n)})$ \cite{CCKK12}.

If all customers are interested in a set of size at most $k$, Balcan and Blum give an $O(k)$-approximation which specializes to a 4-approximation in the case $k = 2$ \cite{BB06}. Amazingly, this may also be tight: building on work by Khadekar {\em et al.} \cite{KKMS09}, Lee showed that there is no $(4-\epsilon)$-approximation when $k=2$ for any constant $\epsilon > 0$ unless the Unique Games conjecture fails \cite{lee15}.

Cheung and Swamy studied the envy-free variant of capacitated pricing problems \cite{CS08}. As mentioned earlier, they show that LP-based
approximations that choose the maximum-profit set of customers for given prices translate to approximation algorithms for envy-free pricing
with capacities while losing an $O(\log \mu_{\max})$-factor. In particular, for envy-free \cgp they get an $O(\log \mu_{\max})$-approximation.

Other variants of envy-free pricing problems have been studied, we do not attempt to comprehensively survey all such models and just sample a few to discuss.
For example, it could be that each customer is interested in acquiring just a single item from
their subset (rather than all items). This was also studied in \cite{GHKKKM05} and follow-up work ({\em eg.} \cite{EFS12}).
Other directions have considered more restricted subsets of items in single-minded pricing,
for example the customers may be interested in the edges of
a subpath of a large path (the {\bf highway problem}) or subpaths of a tree (the {\bf tollbooth problem}). See \cite{GR11} and \cite{GS10}
for definitions and state-of-the-art approximations for these problems.


\subsection{Organization}
Section \ref{sec:prelim} introduces notation and discusses the reduction from \cgp to \lsp. Section \ref{sec:algo}
gives the local search algorithms to prove Theorems \ref{thm:single} and \ref{thm:multi}. The tightness of the analysis
of these local search algorithms is presented in Section \ref{sec:locality}. The randomized LP rounding algorithms
and, ultimately, the proof of Theorem \ref{thm:uniform} appears in Section \ref{sec:lp}. The {\bf APX}-hardness proof appears in \ref{app:hardness}.


\section{Preliminaries}\label{sec:prelim}

We consider graphs that may have parallel edges, unless we explicitly specify otherwise. We do not consider loops. It is easy to extend our
result to cases where some customers may only be interested a singleton bundle. A brief discussion of this can be found in Appendix \ref{app:loops}.

For a set of nodes $S$ in a graph $G = (V,E)$, we let $N(S)$ denote all nodes not in $S$ that are neighbours of some node in $S$.
For $u \in V$ we let $\delta_G(u)$ be all edges having $u$ as an endpoint. Often the subscript $G$ is omitted when
it is clear from the context.  For a subset of edges $B$, we let $\delta_B(u) = \delta(u) \cap B$, again when the graph $G$ is clear from the context.

Again, we sometimes refer to an edge $e$ by $uv$ where $u,v$ are the endpoints of $e$. For brevity, we may use notation like $e = uv \in E$ when we want to consider an edge $e \in E$
but also want to name the endpoints $u,v$ of $e$ as well. The reason for using this notation rather than simply saying $uv \in E$
is that our local search algorithms do work for graphs with parallel edges (i.e. customers interested in identical bundles),
so $e$ would be one particular customer and $u,v$ would name the items that $e$ is interested in.

Given a function $f : T \rightarrow \mathbb R$ on some finite set $T$, for any $S \subseteq T$ we let $f(S)$ denote $\sum_{x \in S} f(x)$.
Similarly, if $p : V \rightarrow \mathbb R_{\geq 0}$ is a pricing of the vertices of a graph $G = (V,E)$, for an edge $e = uv \in E$
we let $p(e)$ denote $p(u) + p(v)$.
For two pricings $p, p' : V \rightarrow \mathbb R_{\geq 0}$ of the nodes of a graph, we let $\texttt{HW}(p, p') = |\{v \in V : p(v) \neq p'(v)\}|$, the number of vertices with different prices between $p$ and $p'$.

Finally, consider an instance $G = (L \cup R, E)$ of \lsp where edges have budgets $b_e$ and vertices have capacities $\mu_v$.
For any pricing $p$ of the vertices, let $\texttt{val}(p) = \max_{F \subseteq E \\ \texttt{ and } (p,F) \texttt{ feasible}} \sum_{e \in F} p(e)$
be the maximum profit of a feasible solution with prices $p$. Note that $\texttt{val}(p)$ can be computed in polynomial time as it is merely
asking for a maximum-weight $\mu$-matching solution using only edges $e = uv$ with $p(e) \leq b_e$ (the weight of such an edge being $p(e)$).


\subsection{Reduction to L-Pricing}
To begin, we use a reduction by Balcan and Blum \cite{BB06} which was stated originally only for the uncapacitated case. We sketch the proof for completeness.
\begin{lemma}[Balcan and Blum \cite{BB06}]\label{lem:reduce}
If there is an $\alpha$-approximation for \lsp then there is a $4\alpha$-approximation problem for \cgp.
\end{lemma}
\begin{proof}
Let $G = (V,E)$ be an instance of \cgp with capacities $\mu$ and budgets $b$. Randomly form $L$ by including each vertex independently
with probability $1/2$ and set $R = V-L$. Discard all edges
with both endpoints in the same set of the partition. Consider an optimum pricing $p^*$ for $G$ and corresponding set of customers $F^* \subseteq E$. Let $F'$ be the restriction of $F^*$ to edges
$e$ with endpoints in each of $L$ and $R$ and consider prices $p'$ where $p'(u) = p^*(u)$ for $u \in L$ and $p'(v) = 0$ for $v \in R$. One can easily check
${\bf E}\left[\sum_{e \in F'} p'(e)\right] = \frac{1}{4} \cdot \sum_{e \in F^*} p^*(e)$.
\end{proof}
This can be efficiently derandomized because we only require pairwise independence of the events $u \in L$ for various $u \in V$, see \cite{randbook} for details behind this technique.


\section{Local-Search Algorithms}\label{sec:algo}

We consider local-search algorithms for \lsp. Recall we are given a bipartite
graph $G = (L \cup R, E)$ where each $v \in L \cup R$ has a capacity $\mu_v \geq 0$, each $e \in E$ has a budget $b_e$,
and we are restricted to setting $p(v) = 0$ for each $v \in R$.

It is clear that there is an optimal solution $p$ such that for each $u \in L$ we have $p(u) = b_e$ for some $e \in \delta(u)$.
Otherwise we could
increase $p(u)$ to the next budget of an edge touching $u$ (or decrease, if $p(u)$ exceeds all budgets of edges touching $u$)
while not decreasing the value of the solution.
So, for $u \in L$ we define $P_u = \{b_e : e \in \delta(u)\}$ to be the set of budgets
of customers interested in item $u$.

We run a local-search approximation based on this observation. Here, a vector $p$ over $L$ is a {\bf pricing} if $p(u) \in P_u$ for each $u \in L$. The local-search algorithm iteratively tries to improve a pricing by changing the price of only one vertex until no such
improvement is possible. The full algorithm is presented in Algorithm \ref{alg:single}. Because a price $p(u)$ is chosen from $P_u$
for each $u \in L$, it is clear that an iteration can be executed in polynomial time.

\begin{algorithm}[H]
\begin{algorithmic}
\STATE let $p$ be any pricing
\WHILE{$\texttt{val}(p') > \texttt{val}(p)$ for some pricing $p'$ with $\texttt{HW}(p, p') = 1$}
\STATE $p \leftarrow p'$
\ENDWHILE
\RETURN $p$
\end{algorithmic}
\caption{Single-Swap Algorithm for \lsp.}
\label{alg:single}
\end{algorithm}

Call a pricing $p$ {\em locally optimal} if it cannot be improved by changing the price for any $u \in L$, note Algorithm \ref{alg:single}
returns a locally-optimal pricing. As is common in local search, we analyze the quality of a locally-optimal solution. In the next subsection we show $\texttt{val}(p) \geq \texttt{val}(p^*)/2$ where $p^*$ is an optimal pricing for the \lsp instance.

The main concern is then the efficiency of the algorithm. Clearly each iteration can be executed in polynomial time but the number of iterations is not apparently bounded. We use a more recent observation from \cite{FKS} to find a solution which may not be a local optimum but is still guaranteed to have value at least $1/2$ of the optimum value (no $\eps$-loss in the guarantee, like one would expect using older methods). The simple application of this trick is discussed in Appendix \ref{app:faster}.


\subsection{Single-Swap Analysis}

We fix $p^*$ to be some particular optimal pricing.
\begin{theorem}\label{thm:localsingle}
For any locally-optimal pricing $p$, $\texttt{val}(p) \geq \texttt{val}(p^*)/2$.
\end{theorem}
\begin{proof}
Let $B \subseteq E$ be the edges that are bought in the local optimum solution, and $B^* \subseteq E$ the edges
that are bought in the global optimum solution. Thus, $\texttt{val}(p) = \sum_{u \in L} p(u) \cdot |\delta_B(u)|$
and $p(e) \leq b_e$ for each $e \in \delta_B(u)$.

For each $u \in L$, consider the local search step that changes the price of $u$ from $p(u)$ to $p^*(u)$. That is, consider $p^u$ where $p^u(u) = p^*(u)$ and $p^u(u') = p(u')$ for $u' \in L-\{u\}$. We refer to this swap as the $p \rightarrow p^u$ swap. For brevity, let $\Delta_u := \texttt{val}(p^u) - \texttt{val}(p)$ and note $\Delta_u \leq 0$ because $p$ is a local optimum. We provide a lower bound on $\Delta_u$ in a way that relates part of the global optimum with part of the local optimum.

First, construct a subset $B' \subseteq B^*$ and an injective mapping $\sigma : B' \rightarrow B$ iteratively as follows in Algorithm \ref{alg:pair}. Intuitively, it greedily pairs some edges in $B^*$ with edges in $B$ sharing the same endpoint in $R$ until no more pairs can be made.
After this pairing, for each $v \in R$ we either have $\delta_{B^*}(v) \subseteq B'$ or $\delta_{B}(v) \subseteq \sigma(B')$ (or both).

\begin{algorithm}[H]
\begin{algorithmic}
\STATE $B' := \emptyset$
\FOR{each $e^* = uv \in B^*$ where $v \in R$}
\IF{there is some $e \in \delta_B(v)$ such that no $e' \in B^*$ has $\sigma(e') = e$}
\STATE {\bf set} $B' := B' \cup \{e^*\}$ and $\sigma(e^*) := e$
\ENDIF
\ENDFOR
\end{algorithmic}
\caption{Constructing $B'$ and $\sigma$.}
\label{alg:pair}
\end{algorithm}

Now we bound $\Delta_u$. One possible matching with the modified prices $p^u$ is
\[ B^u := B \cup \delta_{B^*}(u) - \delta_B(u) - \{\sigma(e) : e \in \delta_{B'}(u) \}. \]
A simple inspection of the definition of $B'$ and $\sigma$ shows this is feasible. That is, it alters $B$ by swapping $\delta_B(u)$ for $\delta_{B^*}(u)$ and removes edges paired, via $\sigma$, with $\delta_{B^*}(u)$
to make room across nodes in $R$ for these new edges. It could be that some edges in $\delta_{B^*}(u)$ are not paired by $\sigma$ but this indicates their right-endpoints already have enough room already to accommodate these edges without removing other edges from $B$. So, $B^u$ respects the vertex capacities.

Now, $\Delta_u$ represents the cost change when using the maximum value matching with the new profits. This can be bounded as follows,
based on the fact that $B^u$ is a feasible solution under prices $p^u$:
\[ 0 \geq \Delta_u \geq p^*(u) \cdot |\delta_{B^*}(u)| - p(u) \cdot |\delta_B(u)| - \sum_{e' \in \delta_{B'}(u)} p(\sigma(e')). \]
Summing over all $u \in L$ and noting each $e \in B$ has its corresponding term appearing in the last sum for at most one $u \in L$
because $\sigma'$ is one-to-one shows $0 \geq \texttt{val}(p^*) - 2 \cdot \texttt{val}(p)$.
\end{proof}


\subsection{An Improved Multi-Swap Algorithm for Bounded Capacities}
Here we consider the restriction of \lsp to instances where $\mu_u \leq C$ for each $u \in L$ for some fixed constant $C$.\footnote{\cgp with $C=1$ is equivalent to maximum-weight matching. We can easily set prices to get the full budget of all edges from any matching.} Note we do not require capacities of $v \in R$ to be bounded.

Let $d \geq 1$ be a fixed integer: larger $d$ will result in better approximation guarantees with a slower, but still polynomial-time, algorithm.
The multi-swap algorithm we consider is given in Algorithm \ref{alg:multi}. Let $\rho = 1 + C + C^2 + \ldots + C^d =  \frac{C^{d+1}-1}{C-1}$. An iteration
runs in polynomial time because $\rho$ is a constant.
\begin{algorithm}[H]
\begin{algorithmic}
\STATE let $p$ be any pricing
\WHILE{there is a pricing $p'$ with $\texttt{HW}(p, p') \leq \rho$ and $\texttt{val}(p') > \texttt{val}(p)$}
\STATE $p \leftarrow p'$
\ENDWHILE
\RETURN $p$
\end{algorithmic}
\caption{Multi-Swap Algorithm For \lsp.}
\label{alg:multi}
\end{algorithm}
As before, call a pricing $p$ {\em locally optimal} if $\texttt{val}(p') \leq \texttt{val}(p)$ for any pricing $p'$ with $\texttt{HW}(p, p') \leq \rho$.
Recall $P_u$ for $u \in L$ is the set of distinct budgets of the edges incident to $u$ and that, in \lsp, we can assume any pricing $p$ has $p(u) \in P_u$ for all $u \in L$.
So, as $C$ and $d$ are constants, a single iteration can be executed in polynomial time by trying all subsets $S \subseteq L$ of bounded size and, for each of those, trying all $\prod_{u \in S} (|P_u|-1) \leq |E|^{O(1)}$ ways to change the prices of all $u \in S$.
We prove the following approximation guarantee.
\begin{theorem}\label{thm:multigood}
Let $p$ be a locally-optimal solution and $p^*$ a global optimum solution. Then $\texttt{val}(p) \geq \frac{C - C^{-d}}{2C - 1 - C^{-d}} \cdot \texttt{val}(p^*)$.
\end{theorem}
So for any fixed $C \geq 2$ and $\eps > 0$ we can choose $d$ large enough to get a $\left(\frac{C}{2C-1} - \eps\right)$-approximation for instances of \lsp where all capacities of nodes in $L$ are bounded by $C$. We can use the trick from Appendix \ref{app:faster} to ensure the number of iterations is polynomially-bounded.

We will soon prove Theorem \ref{thm:cover} stated in Section \ref{sec:results}.
For now, we show how to complete the local search analysis using this result. Let $p^*$ denote an optimal pricing, $B \subseteq E$ the edges
bought in the local optimum $p$, and $B^* \subseteq E$ the edges bought under $p^*$. Let $\sigma : B' \rightarrow B$ be a pairing
constructed in the same way as in the single swap analysis (using Algorithm \ref{alg:pair}) where $B' \subseteq B^*$.

To describe the swaps used in the analysis, first
consider the following auxiliary directed graph $H = (L, F)$ whose nodes are the same as the left-side of this \lsp instance and
whose edges are given as follows. For any $e^* = uv \in B'$, let $w \in L$ be the left-endpoint of $\sigma(e^*)$. Add a directed edge from $u$ to $w$ in $F$.

Observe that both the indegree and outdegree of a vertex in $H$ is at most $C$ by this construction, so Theorem \ref{thm:cover} applies.
Let $\tau : L \times \{0, 1, \ldots, d\}$ be the given weighting of directed balls in $H$. These weights will be used to combine inequalities generated by the test swaps below.

~

~


\noindent
{\bf Choosing the Test Swaps}\\
For any $u \in L$ and any $0 \leq i \leq d$, consider the prices $p^{u,i}$ defined by
\[ p^{u,i}(v) = \left\{\begin{array}{rl}
p^*(v) & \texttt{ if } d_H(u,v) \leq i \\
p(v) & \texttt{ otherwise}
\end{array}\right. \]
Note $\texttt{HW}(p, p^{u,i}) = |B^+(u,i)| \leq C^0 + C^1 + \ldots + C^i \leq \rho$ because the outdegree of each vertex is at most $C$, so $p \rightarrow p^{u,i}$ is a valid test swap.
Let $\Delta_{u,i} =  \texttt{val}(p^{u,i}) -  \texttt{val}(p)$ and note $\Delta_{u,i} \leq 0$ by local optimality.
We bound the difference by explicitly describing a feasible set of edges $B^{u,i}$, namely:
\[ B^{u,i} = B \cup \delta_{B^*}(B^+(u,i)) - \delta_B(B^+(u,i)) - \sigma(\delta_{B'}(\partial B^+(u,i))). \]

That is, add all edges from $B^*$ touching a vertex in the directed ball $B^+(u,i)$ and remove all edge from $B$ that either touch $B^+(u,i)$ or are paired (via $\sigma$)
with an edge in $B'$ that touches $\partial B^+(u,i)$. It is again easy to check that $(p^{u,i}, B^{u,i})$ is a feasible solution: across $u \in L$ we simply exchanged edges in $B$ touching $U$
for edges in $B^*$ touching $u$ and we ensured any new $e^* \in B'$ has $\sigma(e^*)$ removed to make room for $e^*$ across its right-endpoint.
Observe for any $e^* \in \delta_{B'}(B^+(u,i-1))$ that $\sigma(e^*)$ is already removed when $\delta_B(B^+(u,i))$ is removed from $B$, which is why the last part of the definition of $B^{u,i}$
only uses the boundary $\partial B^+(u,i)$ instead of all of $B^+(u,i)$ to remove the remaining edges of $B$ that are paired with $\delta_{B'}(B^+(u,i))$.

Weighting the inequalities by the values $\tau(u,i)$ from Theorem \ref{thm:cover},
\begin{eqnarray}
0 & \geq & \tau(u,i) \cdot \Delta_{u,i} \geq \tau(u,i) \cdot \left(\sum_{e \in B^{u,i}} p^{u,i}(e) - \sum_{e \in B} p(e)\right) \nonumber \\
& = & \tau(u,i) \cdot \sum_{e \in B^* \cap B^{u,i} } p^*(e) - \tau(u,i) \cdot \sum_{e \in B-B^{u,i}} p(e). \label{eqn:lower} 
\end{eqnarray}

It remains to consider the contribution of each edge in $B^*$ and $B$ to this bound if we sum over all $u \in L, 0 \leq i \leq d$.
Observe an edge $e = vw \in B^*$ is ``swapped in'' in this analysis for the swap $p \rightarrow p^{u,i}$ if and only if $v \in B^+(u,i)$.
So by Theorem \ref{thm:cover}, the total contribution of $p^*(e)$ to $\sum_{u,i} \tau(u,i) \cdot \Delta_{u,i}$ is precisely $\frac{C^{d+1}-1}{C-1} \cdot p^*(e)$.

On the other hand, an edge $e = vw \in B$ is ``swapped out'' in this analysis for the swap $p \rightarrow p^{u,i}$ if and only if $v \in B^+(u,i)$ or 
$\sigma^{-1}(e) \in \partial B^+(u,i)$ (if $e$ is indeed paired by $\sigma$).
Again by Theorem \ref{thm:cover}, the total $\tau$-weight of the first event
is exactly $\frac{C^{d+1}-1}{C-1}$ and, if $\sigma^{-1}(e)$ is defined, the total $\tau$-weight of the second event is exactly $C^d$. Thus,
\[ 0 \geq \sum_{u \in L\\ 0 \leq i \leq d} \tau(u,i) \cdot \Delta(u,i) \geq \frac{C^{d+1}-1}{C-1} \cdot \texttt{val}(p^*) - \left(\frac{C^{d+1}-1}{C-1} + C^d\right) \cdot \texttt{val}(p), \]
which proves Theorem \ref{thm:multigood}.


\subsection{Proof of Theorem \ref{thm:cover}}

Before presenting the full proof to conclude the analysis, we consider a simpler setting to develop intuition. Suppose,
for each $0 \leq i \leq d$ and each $u \in L$ there are precisely $C^i$ nodes $w \in L$ with $d_H(w,u) = i$. This would happen if, say, $H$ has indegree and outdegree exactly $C$ at each
vertex and the undirected version of $H$ has girth $> 2d$.
Then setting $\tau(u,i) = 1$ if $i = d$ and 0 otherwise for each $u \in L$ would suffice.

In the general setting without this assumption, we have to consider other directed balls $B^+(u,i)$ for different $0 \leq i \leq d$ and with, perhaps, larger weights than 1.
This is because the radius-$d$ balls $B^+(u,d)$ themselves for various $u \in L$ do not cover each $v \in L$ precisely $\sum_{i=0}^d C^i$ times.

Inductively define $\tau(u,i)$ for $u \in L$ and $0 \leq i \leq d$ as follows:
\[ \tau(u,i) = \left\{\begin{array}{cl}
1 & \texttt{ if } i = d, \\
C^{d-i} - \displaystyle\sum_{j=i+1}^{d} \sum_{v \in L:d_H(v,u) = j-i} \tau(v,j) & \texttt{ otherwise.}
\end{array}\right. \]

The inspiration behind this construction is that in general we would have $d_H(u,v) = i$ for only {\em at most} $C^i$ nodes $u$. So we consider smaller directed balls to make up this deficiency. If we think that the distance $i$ requirement
for each $v \in V$ is exactly $C^i$, then for each $u \in L$ the ball $B^+(u,j)$ contributes to the distance $d-j + d_H(u,v)$ requirement for each $v \in B^+(u,j)$. See Figure \ref{fig:swaps} for an illustration.

The recurrence above ensures the total contribution to the distance $i$ requirement for each $v$ by all all directed balls is exactly $C^i$.  We formalize this idea and show the $\tau$ values are nonnegative
in Lemma \ref{lem:cnt} below.

\begin{figure}[h]
\begin{center}
\includegraphics[width=0.3\textwidth]{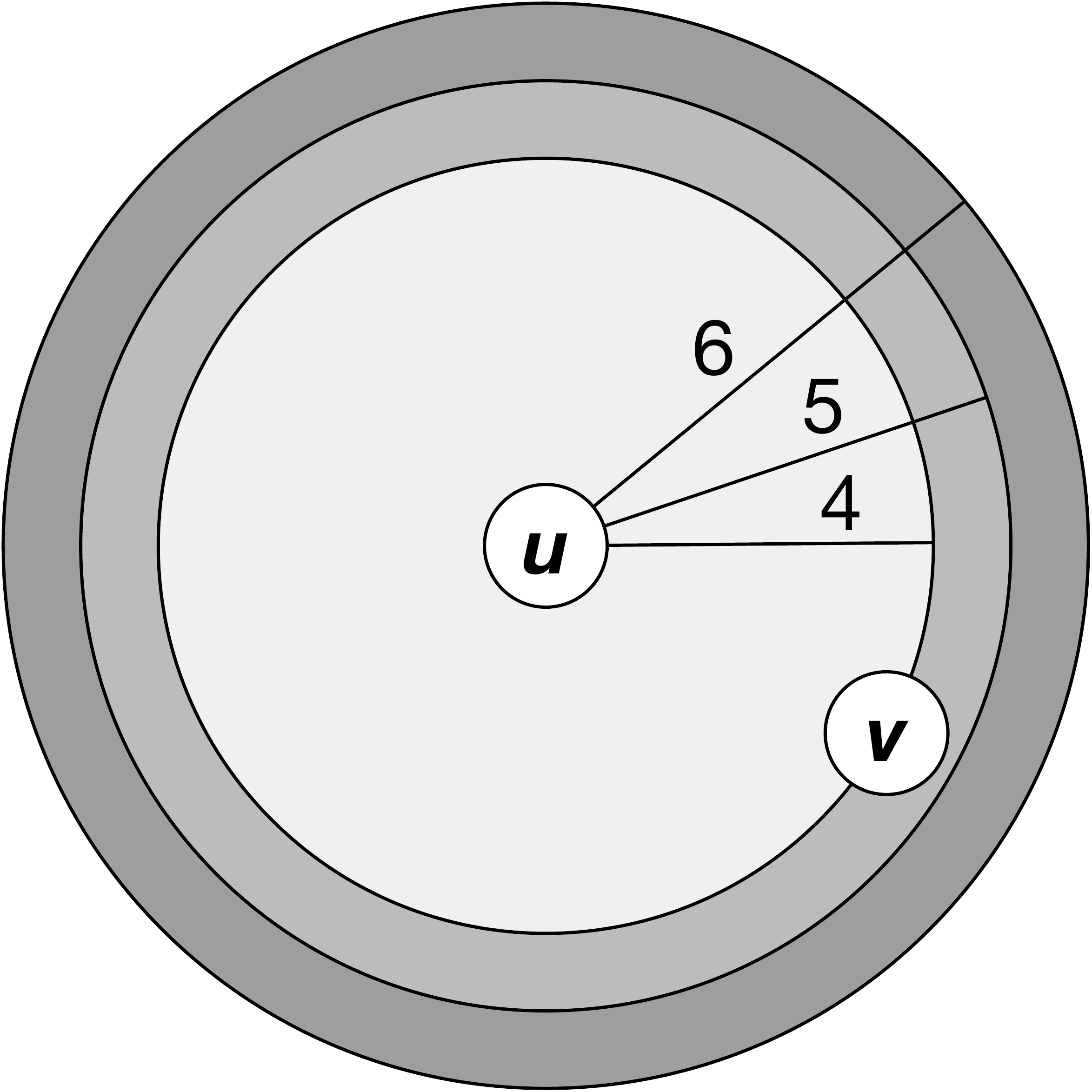}
\end{center}
\caption{Illustration of the case $d = 6$ where $d_H(u,v) = 4$. The directed ball $B^+(u,6)$ contributes to the ``distance 4'' requirement
for $v$, $B^+(u,5)$ contributes to the ``distance 5'' requirement for $v$, and $B^+(u,4)$ contributes to the ``distance 6'' requirement for $v$.}  \label{fig:swaps}
\end{figure}

\begin{lemma}\label{lem:cnt}
For each $u \in L, 0 \leq i \leq d$
we have $\displaystyle\sum_{j=i}^{d} \sum_{v \in L : d_H(v,u) = j-i} \tau(v,j) = C^{d-i}$
and $0 \leq \tau(u,i) \leq C^{d-i}$.
\end{lemma}
\begin{proof}
The equality is by construction and the observation that $d_H(v,u) = 0$ if and only if $v = u$.
The inequalities are proven inductively with the base case $i = d$ being given.
Now suppose for $i < d$ we know $0 \leq \tau(u,j) \leq C^{d-j}$ for any $i < j \leq d$ and any $u \in L$.
By the recurrence for $\tau(u,i)$ and because $\tau(v,j) \geq 0$ for any $i < j \leq d$ and $v \in V$, we see $\tau(u,i) \leq C^{d-i}$.
Next, we prove $\tau(u,i) \geq 0$ for each $u \in L$.

For any $i < j \leq d$ and any $v \in L$ with $d_H(v,u) = j-i$, there is some $w \in L$ such that $d_H(v,w) = i-j-1$ and $d_H(w,u) = 1$.
That is, consider a shortest $v-u$ path $P$ in $H$, as $i < j$, we have $v \neq u$ so 
the second-last node on this path is a node $w$ whose distance to $u$ is 1 (it could be $w=v$, if $j-i = 1$).

From this and using the equality from the first part of the theorem statement, we bound the double sum in the recurrence defining $\tau(u,i)$ by
\begin{eqnarray*}
\sum_{j=i+1}^d  ~ \sum_{v \in L : d_H(v,u) = j-i} \tau(v,j) & \leq & \sum_{w : d_H(w,u) = 1} ~ \sum_{j=i+1}^d ~ \sum_{v \in L : d_H(v,w) = j-(i+1)} \tau(v,j) \\
& = & \sum_{w : d_H(w,u) = 1} C^{d-(i+1)} \\
& \leq & C^{d-i}.
\end{eqnarray*}
The last bound follows as each $v \in L$ has indegree at most $C$ in $H$.
Thus, from the recurrence again, we see $\tau(u,i) \geq 0$.
\end{proof}

Lemma \ref{lem:cnt} finishes the proof of Theorem \ref{thm:cover} as follows. The first bullet point in Theorem \ref{thm:cover} follows by summing over all $0 \leq i \leq d$. The second point
follows by fixing $i = 0$.


\section{Locality Gaps}\label{sec:locality}
For an instance of \lsp and a value $\rho \geq 1$, we say its {\bf locality gap with respect to the $\rho$-swap heuristic} is the maximum of $\frac{\texttt{val}(p^*)}{\texttt{val}(p)}$ over pricings $p$ that cannot be improved by changing the value of up to $\rho$ entries $p(v)$ where $p^*$ is any optimal pricing. That is, the locality gap is the worst possible approximation guarantee of a local optimum for that instance.


\subsection{Single-Swap}
\begin{theorem}
For any $C \geq 1$ and $\eps > 0$, there is an instance $\Phi = (G,\mu,b)$ of \lsp where all $u \in L$ have capacity $C$ and all $v \in R$ have capacity 1 such that the locality gap of $\Phi$ is at least $2 - \eps$ with respect to the single-swap heuristic.
\end{theorem}
Thus, our single-swap analysis is tight. While this is most striking when $C = 1$, we remark it is still interesting for larger $C$ because it is not obvious, {\em a priori}, that the single-swap algorithm's analysis cannot be improved as the capacities in $L$ increase.
\begin{proof}
For $n \geq 2$, consider the graph $G^{n,C} = (L \cup R, E)$, $L = \{u_i : 1 \leq i \leq n\}$ and $R = \{v_{i,j} : 1 \leq i \leq n, 1 \leq j \leq C\}$.
The edges are the union of the edges on the paths $P_j= \{u_1, v_{1,j}, u_2, v_{2,j}, \ldots, u_n, v_{n,j}\}$ for $1 \leq j \leq C$.
We use $\mu_u = C$ for $u \in L$ and $\mu_v = 1$ for $v \in R$.

The budgets are given as follows. First let $E_{LOC} = \{u_iv_{i,j} : 1 \leq i \leq n, 1 \leq j \leq C\}$ and
$E_{OPT} = \{u_iv_{i-1,j} : 2 \leq i \leq n, 1 \leq j \leq C\}$. All edges in $E_{LOC}$ have a budget of 1 and all edges in $E_{OPT}$ have
a budget of 2.

Using $p^*(u) = 2$ for every vertex in $L$ and corresponding edges $E_{OPT}$ is a solution with value of $2C(n-1)$. Now consider the solution with $p(u) = 1$ for each vertex in $L$. This solution is just $E_{LOC}$ with a value of $Cn$, which can be seen to be the optimal matching
under these prices because the capacity of every vertex in $L$ is saturated by $E_{LOC}$. Note $\texttt{val}(p) = \left(\frac{1}{2-2/n}\right)\cdot \texttt{val}(p^*)$.
We claim this is a local optimum with respect to the single-swap heuristic.

The only possible swap is to change the price of some $u_i$ to 2. If $i = 1$, this is clearly not an improving swap because no edge incident to $u_1$ can afford the new price and all other vertices are priced 1 so no matching has value $\geq n-1$. So suppose $i \geq 2$.

The only edges incident to $u_i$ that can afford this new price are $(u_i, v_{i-1,j})$ for all $1 \leq j \leq C$. Furthermore, for any
$\mu$-matching $B$ that does not use an edge $e \in \delta_{E_{OPT}}(u_i)$, we can get a better $\mu$-matching (with respect to the new prices)
by adding $e$ to $B$ and, if necessary, removing some edge of $E_{LOC} \cap B$ sharing the right-endpoint with $e$.

Thus, the optimum matching after changing $p(u)$ to 2 uses all of $\delta_{E_{OPT}} \cap B$. After fixing these edges, which have total value $2C$, it is easy to see the best matching we can get in the graph obtained by removing the endpoints of edges $\delta_{E_{OPT}}(u) \cap B$ (plus all edges incident to these endpoints) has value at most $C(n-2)$. So this is not an improving swap.
\end{proof}


\subsection{Multi-Swap}
Before presenting the lower bound, we describe a construction of a layered graph with high girth and particular degree bounds.
The construction is depicted in Figure \ref{fig:lowerbound}.

\begin{figure}[h]
\includegraphics[width=1.0\textwidth]{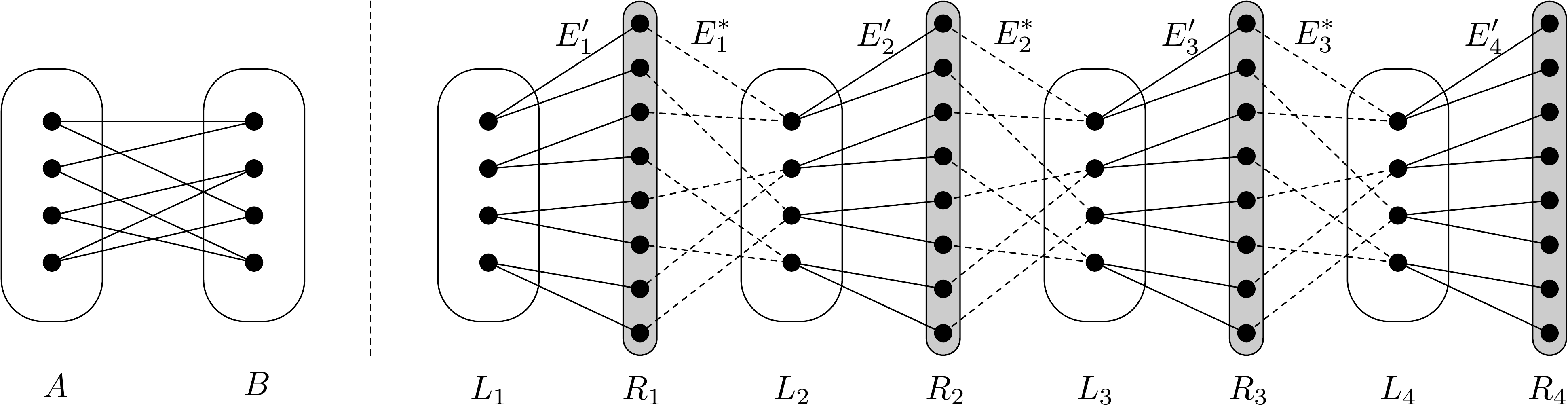}
\caption{An example with $C=2$.
{\bf Left}: the bipartite graph $H$ constructed from $H'$ in the proof of Lemma \ref{lem:girth}.
{\bf Right}: the resulting graph $G^t$ with $t = 4$. Edges of $E'$ are solid and the edges of $E^*$ are dashed.}\label{fig:lowerbound}
\end{figure}

\begin{lemma}\label{lem:girth}
For any $C \geq 2, \rho \geq 1$ and $t \geq 1$ there is a simple, layered, and bipartite graph $G^t = (L \cup R, E)$ with consecutive layers $L_1, R_1, L_2, \ldots, L_t, R_t$ where the subgraph
induced by $L_i$ and $R_i$ is a $(C,1)$-biregular bipartite graph and the subgraph induced by $R_i$ and $L_{i+1}$ is a $(1,C)$-biregular bipartite graph (for each relevant $i$). Further, $G^t$ has girth exceeding $2t \cdot \rho$.
\end{lemma}
Note, this implies $|L_i| = |L|/t$ and $|R_i| = |R|/t = C \cdot |L_i|$ for each $1 \leq i \leq t$.
\begin{proof}
In \cite{S63}, it is is shown that for any $C', g \geq 3$ there is a simple, connected $C'$-regular graph whose girth ({\em i.e.} shortest cycle length) is at least $g$. In our setting, this means a $(2C)$-regular graph with
girth exceeding $\rho \cdot t$ exists where $\rho, t$ are as in the statement of Lemma \ref{lem:girth}.
Call this graph $H = (V,F'')$.

As $H$ is $(2C)$-regular and connected it contains an Eulerian circuit.
Direct all edges along this circuits so each vertex of $H$ has indegree $C$ and outdegree $C$. Finally, build a bipartite graph $H = (A \cup B, F)$ where $A$ and $B$ are
disjoint copies of $V$ where $u \in A$ and $v \in B$ is an edge of $F$ if $uv$ is a directed edge obtained when we directed the Eulerian circuit.

Build the following layered graph $G^t = (L \cup R, E)$. For each $1 \leq i \leq t$, let $L_i$ be a set of size $V$
and $R_i$ be a set of size $|F|$. Recall that both $A$ and $B$ in $H$ are viewed as copies of $V$ in $H'$, so each $L_i$ can be viewed either as a copy of $A$ or as a copy of $B$, when appropriate.
Now, for each $u \in A$, each $e \in \delta_H(u)$, and each layer $1 \leq i \leq t$, add an edge in $G^t$ from the copy of $u$ in $L_i$ to the copy of $e$ in $R_i$. Call the set of all such edges added for a given $i$ $E'_i$.
Similarly, for each $v \in B$, each $e \in \delta_H(u)$, and each layer $1 \leq i \leq t-1$, add an edge in $G^t$ from the copy of $e$ in layer $R_i$ to the copy of $v$ in $L_i$. Call the set of all such edges added for a given $i$ $E^*_i$.

Then let $L = \cup_{i=1}^t L_i$, $R = \cup_{i=1}^t R_i$ and $E = (\bigcup_{i=1}^{t-1} E'_i \cup E^*_i) \cup E'_t$. This construction is depicted in Figure \ref{fig:lowerbound} in Section \ref{sec:locality}.

To complete the analysis, consider a simple cycle $C$ in $G^t$.
 Note that $C$ alternates between using nodes in $L$ and nodes in $R$. Furthermore, if $C$ uses nodes consecutive nodes $a \in L_i, b \in R_j, c \in L_k$ (where $|j-i| \leq 1$ and $|k-j| \leq 1$)
then the nodes of $H$ corresponding to $a$ and $c$ are connected by an edge in $H$ that corresponds to node $b$. Thus, the cycle $C$ corresponds to a circuit $C'$ of $H^t$ with $|C'| = |C|/2$. Here, $C'$ may use
an edge more than once so $|C'|$ measures the steps taken by the circuit $C'$.

Consider any node $a \in C \cap L$ and say it is a copy of node $v$ of $H$. Because the cycle $C$ is simple in $G^t$, then the two adjacent nodes $b, b'$ to $a$ on $C$ correspond to distinct edges in $H$ incident to $v$.
This is true for every $a \in C \cap L$, so the set of nodes of $H$ corresponding to nodes in $C \cap L$ are incident to at least two distinct edges traversed by $C'$. That is, the edges used on the circuit $C'$ contain a cycle.
As the girth of $H$ is at least $\rho \cdot t$, then $|C| = 2 \cdot |C'| \geq 2 \cdot \rho \cdot t$.
\end{proof}

\begin{theorem}
For all $C \geq 2, \rho \geq 1$ and $\epsilon > 0$, there is an instance $\Phi = (G,\mu,b)$ of \lsp where all $u \in L$ have capacity $C$ and all $v \in R$ have capacity 1 such the locality gap if $\Phi$ is at least $\frac{2C-1}{C} - \epsilon$ with respect to the $\rho$-swap heuristic.
\end{theorem}
That is, our bound on the locality gap for the multi-swap algorithm on instances with bounded capacity is tight.
\begin{proof}
Fix $C \geq 2, \rho \geq 1, \epsilon > 0$ and let $t$ be such that $\frac{2C-1}{C} \cdot \frac{t-1}{t} \geq \frac{2C-1}{C} - \epsilon$. Let $G^t = (L \cup R, E)$ be the graph from Lemma \ref{lem:girth} for these parameters $C, \rho, t$.
For each $1 \leq i \leq t$, let $E'_i$ be the edges connecting $L_i$ to $R_i$ and for each $1 \leq i < t$ let $E^*_i$ be the edges connecting $R_i$ to $L_{i+1}$. Naturally, let $E' = \cup_{i=1}^t E'_i$ and $E^* = \cup_{i=1}^{t-1} E^*_i$.
See Figure \ref{fig:lowerbound} for an illustration. Let $n$ be such that $|L_i| = n$ and $|R_i| = C \cdot n$ for all $1 \leq i \leq t$.

The optimum is at least $(2C-1) \cdot C \cdot (|L|-n)$, which can be seen by using edges $E^*$ where each vertex in $L$ has a price of $2C-1$.
Now consider the pricing $p$ that uses price $C$ for each vertex in $L$. The optimum set of edges to buy with these prices is $E'$ with a value of $C \cdot C \cdot |L|$.
The proof of the following claim appears below.
\begin{claim}\label{claim:localopt}
The pricing $p$ is locally optimal with respect to the $\rho$-swap procedure.
\end{claim}
If so, then the locality gap of this instance is as bad as $\frac{2C-1}{C} \cdot \frac{|L|-n}{|L|} = \frac{2C-1}{C} \cdot \frac{t-1}{t}$,
as required.
\end{proof}

\begin{proof}[Proof of Claim \ref{claim:localopt}]

Consider any pricing $p'$ with $\texttt{HW}(p, p') \leq \rho$. Let $X \subseteq L$ be the nodes $v$ with $p(v) \neq p'(v)$.
So $p'(v) = 2C-1$ for $v \in X$. If some vertex $v \in X$ lies in $L_1 \neq \emptyset$ then no edge incident to $v$ can afford the price $2C-1$, so we may assume that $X \cap L_1 = \emptyset$.

We show $\texttt{val}(p') \leq \texttt{val}(p)$. We first claim any optimal set of edges $M^*$ under this price includes all of $\delta_{E^*}(X)$
and excludes all of $\delta_{E'}(X)$. The latter is simple, no edge in $\delta_{E'}(X)$ can afford the price of its endpoint in $L$. Then if any $e \in \delta_{E^*}(X)$
is missing from $M^*$, we can get an even better solution by adding $e$ and removing, if necessary, an edge of $E' \cap M$ sharing its $R$-endpoint with $e$.
The value increases by at least $(2C-1) - C = C-1$.

So $M^*$ contains all edges of $\delta_{E^*}(X)$ with value $2C-1$ each plus some edges in $\delta(L-X)$ (which could be in either $E'$ or $E^*$) with value $C$ each. We then see the value
of $M^*$ is
\begin{equation}\label{eqn:goal}
(2C-1) \cdot C \cdot |X| + C \cdot |\delta_{M^*}(L-X)|.
\end{equation}
The rest of the proof focuses on showing the following:
\begin{equation}\label{eqn:finalbound}
|\delta_{M^*}(L-X)| \leq C \cdot |L| - (2C-1) \cdot |X|.
\end{equation}
If this holds, we can bound \ref{eqn:goal} by $C \cdot C \cdot |L|$ thus showing $\texttt{val}(p') \leq \texttt{val}(p)$.

To show \ref{eqn:finalbound}, first consider the graph $G'$ obtained from $G^t$ by directing all edges to higher layers: so an edge in $E'_i$ is directed from $L_i$ to $R_i$ and an edge in $E^*_i$ is directed from $L_i$ to $R_{i+1}$.
Let $S$ consist of $R_t$ plus all nodes reachable from $X$ in $G'$, including $X$ itself. We claim $|\delta^{in}_{G'}(S)| = C \cdot n$. This can be seen easily:
\[
|\delta^{in}_{G'}(S)| = |\delta^{in}_{G'}(S)| - |\delta^{out}_{G'}(S)|
 =  \sum_{v \in S} |\delta^{in}_{G'}(v)| - \sum_{v \in S} |\delta^{out}_{G'}(S)| 
 =  \sum_{v \in S \cap R_t} |\delta^{in}(v)| = n \cdot C.
\]
The first equality holds because $\delta^{out}_{G'}(S) = \emptyset$ by construction of $S$, the second holds for any cut of any directed graph, and the third holds because $S \cap L_1 = \emptyset$ (as $X \cap L_1 = \emptyset$)
and because every vertex not in $L_1$ or $R_t$ has equal in- and out-degree.

Now let $Y$ be all endpoints of edges in $\delta_{E^*}(X)$ and let $G''$ be the subgraph of $G'$ obtained by deleting $Y$ and incident edges. Let $S' = S-Y$, we claim $|\delta_{G''}(S')| \leq n \cdot C - (C-1) \cdot |X|$.
One should think that $\delta_{G''}(S')$ is obtained by deleting edges of $\delta^{in}_{G'}(X) \cap \delta^{in}_{G'}(S)$ from $\delta^{in}_{G'}(S)$. There are precisely $C \cdot |X|$ edges in $\delta^{in}_{G'}(X)$, we show at least $(C-1) \cdot |X|$ of there were
also in $\delta^{in}_{G'}(S)$.

To that end, consider an edge $e \in \delta^{in}_{G'}(x)$ for some $x \in X$ that does not lie in $\delta^{in}_{G'}(S)$. Then $v$ is reachable from some other node of $X$ in $G'$ by construction of $S$, pick the deepest such node
and call this node $\tau(e)$. By this choice for $\tau(e)$, there is a $\tau(e)-x$ path in $G'$ that avoids every other vertex in $X$. Also, the length of this path is at most $2t$ because the paths are monotone with respect
to the layers of $G'$. Also note for two different $e, e' \in \delta^{in}_{G'}(x) - \delta^{in}_{G'}(S)$ that $\tau(e) \neq \tau(e')$, or else we have two different $\tau(e)-x$ walks implying there is a cycle of length at most $4t$ in $G^t$
which is not possible.

Build an auxiliary graph $\mathcal T = (X, F)$ where for each $e \in \delta^{in}_{G'}(x) - \delta^{in}_{G'}(S)$ for some $x \in X$ we include an undirected edge from $\tau(e)$ to $x$ in $F$.
By the above discussion, this is a simple graph. We also claim it is a forest, otherwise consider a cycle $C$ in $\mathcal T$. Focus on some edge $xy \in C$ and let $z \notin \{x,y\}$ be another node in $C$. As the $xy$-path from the construction
in the last paragraph avoids $z$, we get two different $x-y$ walks in $G^t$ by following the paths corresponding to the two directions around $C$ from $x$ to $y$. But this is impossible because $G^t$ has no cycle of length at most $2t \cdot |X| \leq 2t \cdot \rho$. So, $|F| \leq |X|-1$ meaning $|\delta^{in}_{G'}(X) - \delta^{in}_{G'}(S)| \leq |X|-1$. Thus,
\begin{equation}\label{eqn:cutbound}
|\delta^{in}_{G''}(S')| \leq C \cdot n - (C-1) \cdot |X|.
\end{equation}

Now we can prove \ref{eqn:finalbound}. Let $\overline{G}''$ be the undirected version of $G''$, so $\overline{G''}$ is obtained from $G^t$ by deleting $Y$ and its incident edges from $G^t$. Call a subset of edges of $\overline{G}''$ a matching
if they satisfy the capacity constraints of nodes in $\overline{G}''$. Note $\delta_{M^*}(L-X)$ is a matching.

We bound the size of a maximum matching in $\overline{G}''$. First, observe $M := E^* - \delta(X)$ is a matching in $\overline{G}''$ and that $G''$ is the directed graph we get by directing edges along this matching. That is, the set of $L_1-R_t$ paths in $G''$ are exactly the set of $M$-alternating path. By the max-flow/min-cut theorem, the maximum number of edge-disjoint $M$-alternating paths is at most $|\delta^{in}_{G''}(S')| \leq C \cdot n - (C-1) \cdot |X|$. So the maximum size of a matching
in $\overline{G}''$ is at most
\[ |M| + C \cdot n - (C-1)\cdot |X| = C \cdot (L-n) - C \cdot |X| + C \cdot n - (C-1) \cdot |X| \leq C \cdot |L| - (2C-1) \cdot |X|. \]
This proves \ref{eqn:finalbound} and completes the analysis of the locality gap.

\end{proof}


\section{LP-Based Approximations}\label{sec:lp}

So far, our focus has been on approximations based on local search. Here, we consider linear programming relaxations for \lsp. Recall for each $u \in L$ that $P_u = \{b_e : e \in \delta(u)\}$
is a set of possible prices for vertex $u$: there is an optimal solution that selects $p(u)$ from $P_u$ for each $u \in L$.

For $u \in L$ and $p \in P_u$, we let $y_{u,p}$ be a variable indicating we select price $p$ for $u$. Similarly, for each $e = uv \in E$ and $p \in P_u$ such that $p \leq b_e$, we let $x_{e,p}$ be a variable indicating 
edge $e$ is selected and vertex $u$ is assigned price $p$ (so $e$ buys their bundle at price $p$). The following relaxation provides an upper bound on the optimal solution to the given instance of the \lsp.

\begin{alignat}{3}
\texttt{\bf maximize:} \quad& \sum_{e = uv} \sum_{p \in P_u} p \cdot x_{e,p}& \tag{LP-Pricing} \label{lp:pricing} \\
\texttt{\bf subject to:} \quad& \sum_{p \in P_u} y_{u,p} \quad &= &\quad 1 && \quad  \forall~u \in L \label{lp:oneprice}\\
& \sum_{e \in \delta(u)} x_{e,p} \quad &\leq &\quad y_{u,p} \cdot \mu_u && \quad \forall~u \in L, p \in P_u \label{lp:capL} \\
& \sum_{e=uv \in \delta(v)} \sum_{p \in P_u} x_{e,p} \quad &\leq& \quad \mu_v  && \quad \forall~v \in R \label{lp:capR} \\
& x_{e,p} \quad &\leq &\quad y_{u,p} && \quad \forall~u \in L, e \in \delta(u), p \in P_u \texttt{ s.t. } p \leq b_e \label{lp:cover} \\
& x, y \quad & \geq &\quad 0 && \quad \nonumber
\end{alignat}

Constraints \ref{lp:oneprice} indicate one price must be selected for each $u \in L$, \ref{lp:capL} ensures the capacity constraints for $u \in L$ are satisfied and \ref{lp:capR} ensures the capacity constraints
for $v \in R$ are satisfied, \ref{lp:cover} ensures we must set the price of $u$ to $p$ if we are to have $e$ pay $p$.


\subsection{Randomized Rounding Algorithms}
We first show in simple graphs with large capacities for $R$, the integrality gap is close to 1.
\begin{theorem} \label{thm:LPBigCap}
For $\eps > 0$, the integrality gap of \eqref{lp:pricing} is $1-2\eps$ in simple graphs when $\mu_v \geq 3 \ln(1/\eps)/\eps^2+1$ for all $v \in R$.
\end{theorem}
\begin{proof}
Consider the following randomized rounding algorithm. For each $u \in L$, sample a price $p'(u) \in P_u$ from the distribution with ${\bf Pr}[p'(u) = p] = y_{u,p}$.
This is a distribution by \ref{lp:oneprice} and nonnegativity of $y$. For brevity, we will let $p'(e) = p'(u)$ for an edge $e = uv$.

Then define a fractional matching in $G$ as follows. The idea is that we want to assign a value of $x_{e,p'(e)}/y_{u,p'(e)}$ to each edge,
this is at most 1 by \ref{lp:cover}. By \ref{lp:capL} this fractional matching
would always satisfy the capacity constraints for nodes in $L$. But it may violate constraints for nodes in $R$. The obvious solution would be to scale each of these fractional values to be a feasible matching
satisfying all vertex constraints. We take a simpler view which is sufficient for our purposes, we scale all resulting values by $1-\eps$, and then outright discard edges $e = uv$ where the capacity of $v$ is still
violated after this scaling.

More precisely, for each $e = uv$ we first let $x''_e = (1-\eps)\cdot \frac{x_{e,p'(e)}}{y_{u,p'(e)}}$ (using 0 if $p'(e) > b_e$). Then for each edge $e = uv$, we define
\[ x'_e = \left\{ \begin{array}{rl}
x''_e & \texttt{if } \sum_{e' = u'v \in \delta(v)} x''_e \leq \mu_v, \\
0 & \texttt{otherwise}.
\end{array}\right. \]
Now $x'(\delta(w)) \leq \mu_w$ for each $w \in L \cup R$.
So, by integrality of the bipartite $\mu$-matching polytope, $\texttt{val}(p') \geq \sum_{e} p \cdot x'_e$ as there is an integral matching obtaining
at least as much value as the fractional matching $x'$.

For any $e = uv \in E$ let $\mathcal B_e$ be the {\em bad event} that $\sum_{e' = u'v \in \delta(v), e' \neq e} x''_{e'} \geq \mu_v-1$.
Notice that the second case in the definition of $x'_e$ applies only if event $B_e$ happens. We show
$\Pr[\mathcal B_e] \leq \eps$. If so, for each $e = uv \in E$ the fact that $\mathcal B_e$ is independent of the choice of $p'(e)$ (as $G$ is a simple graph) we then have 
\[ {\bf E}[p'(e) \cdot x'_e] \geq (1-{\bf Pr}[\mathcal B_e]) \cdot (1-\eps) \cdot {\bf E}\left[p'(e) \cdot \frac{x_{e,p'(e)}}{y_{u,p'(e)}}\right] \geq (1-\eps)^2 \cdot \sum_{p \in P_u} p \cdot x_{e,p}.  \]
Summing over all edges, ${\bf E}\left[\sum_{e \in E} p'(e) \cdot x'_e\right] \geq (1-2\eps) \cdot \sum_{e = uv} \sum_{p \in P_u} p \cdot x_{e,p}$, as required.

To bound ${\bf Pr}[\mathcal B_e]$,
for an edge $e'$ let $X_{e'}$ denote the random variable with value $(1-\eps) \cdot x_{e',p'(u)}/y_{u',p'(e')}$ and let $X^e = \sum_{e' \in \delta(v), e' \neq e} X_{e'}$.
Then ${\bf E}[X_{e'}] = (1-\eps) \cdot \sum_{p \in P_{u'}} x_{e',p}$ so by \ref{lp:capR} we have ${\bf E}[X^e] \leq (1-\eps) \cdot \mu_v$.

Again by simplicity of $G$, the random variables $X_{e'}$ are independent for different $e' \in \delta(v), e' \neq e$.
Let $Y = (1-\eps) \cdot (\mu_v-1)$. As $Y \geq {\bf E}[X^e]$ and $0 < \eps < 1$, by a standard Chernoff bound ({\em eg}. Theorem 1.1 in \cite{dpbook}) we have
${\bf Pr}[X > (1+\eps) \cdot Y] \leq \exp(-Y \eps^2/3) \leq \eps$.
Finally, since event $\mathcal B_e$ implies $X^e \geq \mu_v-1 \geq (1+\eps) \cdot Y$, we have ${\bf Pr}[B_e] \leq \eps$, as required.
\end{proof}
The fact that $G$ was simple was used in the application of the Chernoff bound. The random variables $X_{e'}$ for edges in $\delta(v)$ for some $v \in R$ are independent if $G$ is simple.

In fact, this proof generalizes to providing a constant bound on the integrality gap for any instance, even if capacities are small or there are parallel edges. Simply modify the proof to first set $x''_e = \frac{1}{2} \cdot \frac{x_{e,p'(e)}}{y_{u,p'(e)}}$ and set $x'$ as before. Instead of Chernoff bounds, just use Markov's inequality (which does not require independence, thus does not require $G$ to be simple) to show ${\bf Pr}[\sum_{e' = u'v \in \delta(v)} x''_e \leq \mu_v] \geq \frac{1}{2}$. Thus, ${\bf E}[x'_e] \geq \frac{x_e}{4}$.
\begin{lemma}
The integrality gap of \eqref{lp:pricing} no worse than $1/4$ in any instance of \lsp.
\end{lemma}
The choice of $1/2$ in the definition of $x''_e$ is optimal for this analysis. Note, this approximation guarantee is even worse
than our single-swap algorithm.


\subsection{Proof of Theorem \ref{thm:uniform}}\label{app:uniform}
Here we combine the results from Theorem \ref{thm:multi} and \ref{thm:rounding} (or Theorem \ref{thm:LPBigCap}) to provide a better than $2$-approximation for the instances with uniform capacities, where all vertices $w \in L \cup R$ have capacity $\mu_w = C$. 

This is obtained using a slightly more refined analysis of the randomized rounding procedure. We used simpler Chernoff bounds in the proof of Theorem \ref{thm:LPBigCap} to keep the bound simpler to state since the main the goal was to show the guarantee for \lsp approaches 1 as $C$ increases. But since we are interested in optimal constants at this point, we analyze a tighter Chernoff bound.

\begin{theorem}
The randomized rounding procedure produces a solution for \lsp with expected profit at least $0.516 \cdot OPT_{LP}$ in simple graphs where all capacities are at least 22.
\end{theorem}
\begin{proof}
Let $\alpha := 0.57$. In the analysis of the randomized rounding procedure, we explicitly constructed a fractional matching after sampling all of the prices. In this proof, we let $x'_e = \alpha \cdot x_{e, p'(e)}/y_{u, p'(e)}$ for each edge $e$ after sampling $p'$.

For each $v \in R$, each $e \in \delta(v)$ has $x'_e$ being a random value between 0 and 1. Set $\epsilon_v := \frac{c-1}{\alpha \cdot c} - 1$, so that $(1+\epsilon_v) \cdot \alpha \cdot \mu_v = \mu_v-1$.
Since these are independent for different $e$ and since the expected value of $x'(\delta(v))$ is at most $\alpha \cdot \mu_v$, then by a Chernoff bound we have
for any $e \in \delta(v)$ that
\[ {\bf Pr}[x(\delta(v)-\{e\}) > \mu_v - 1] \leq {\bf Pr}[x(\delta(v)-\{e\} \geq (1+\epsilon_v) \cdot \alpha \cdot \mu_v] \leq \left(\frac{e^{\epsilon_v}}{(1+\epsilon_v)^{1+\epsilon_v}}\right)^{\alpha \cdot \mu_v}. \]
Now, as a function of $\mu_v$ the right hand side increases as $\mu_v$ decreases. Since we assumed $\mu_v \geq 22$, we evaluate this at $\mu_v = 22$ to get for any $\mu_v \geq 22$ that
\[  {\bf Pr}[x(\delta(v)-\{e\}) > \mu_v - 1] \leq 0.0937076. \]
Continuing as in the proof of Theorem \ref{thm:LPBigCap}, the expected value of $x''_e$ for $e = uv$ when $p'(e)$ is given is at least
\[ \frac{x_{e, p'(e)/y_{v, p'(e)}}}{\alpha} \cdot \left(1- {\bf Pr}[x(\delta(v)-\{e\}) > \mu_v - 1]\right) \geq \frac{x_{e, p'(e}}{y_{v, p'(e)}} \cdot 0.516. \]
Therefore, the expected contribution of $e$'s value to the matching $x''$ over the random choice of $p'$ is at least $0.516 \cdot \sum_{p \in P_u} p \cdot x_{e,p}$. Summing over all $e$ completes the proof.
\end{proof}

We can now finish the main proof from this section.
\begin{proof}[Proof of Theorem \ref{thm:uniform}]
If $C \leq 21$, use the multiswap local search algorithm to get a solution with profit $\geq \left(\frac{21}{41} - \epsilon\right) \cdot OPT$ for \lsp. If $C \geq 22$, use the randomized rounding procedure to get a solution whose cost is at least $0.516 \cdot OPT$.
For small enough $\epsilon$, $0.516 > \frac{21}{41} - \epsilon$, so in either case we get profit at least $\left(\frac{21}{41} - \epsilon\right) \cdot OPT$. In terms of approximation guarantees, this yields an approximation guarantee of at most $1.952381$ (again, for small enough $\epsilon$). Using Lemma \ref{lem:reduce}, we get a 7.8096-approximation for \cgp.
\end{proof}

Note, our analysis here may still not be optimal: for example, one could consider a smoother scaling of $x''$ instead of simply discarding edges incident to some $v \in R$ whose capacity is violated and, perhaps, get a smaller ``threshold'' value for $C$ (smaller than 22) for which the randomized rounding outperforms the multi-swap algorithm.


\section{{\bf APX}-hardness for \lsp}\label{app:hardness}
\begin{theorem}
\lsp is {\bf APX}-hard, even if all capacities are at most 4 and all customers have a budget of 1 or 2.
\end{theorem}

\begin{proof}
We reduce from the \textsc{Vertex Cover} problem for 3-regular graphs, which is known to be {\bf APX}-hard \cite{AK00}.
Let $G = (V,E)$ be a 3-regular graph, with $|V| = n$ nodes and $|E| = m = \frac{3n}{2}$ edges.

Construct the following bipartite graph $G' = (L \cup R, E')$ from $G$.
Here, $L$ is a copy of $V$ and $R$ is a copy of $V$ plus a copy of $E$.
Each $v \in L$ has capacity 4 and each vertex in $R$ has a capacity of 1.
For a node $v_i \in V$, let $l_i$ denote its copy in $L$ and $r_i$ denote its copy in $R$. Similarly, for each edge $e_j \in E$ let $d_j$
denote its copy in $R$.

All customers have budget equal to $1$ or $2$, and they fall into two classes: node customers and edge customers. For each $v_i \in V$, we have a node customer who is interested in $l_i$ and $r_i$ with budget $2$. For each edge $e_j = v_iv_k \in E$, we define two edge customers interested in $l_id_j$ and $l_kd_j$ respectively, both with budget $1$.
We claim that the optimal solution to this \lsp instance $G'$ has profit $m + 2n - k$ where $k$ is the size of the smallest vertex cover of $G$.

First, suppose $S$ is a vertex cover of $G$ with  $|S| = k$. Consider the pricing $p$ with $p(l_i) = 1$ if $l_i \in S$ and $p(l_i) = 2$ if $l_i \notin S$. As $S$ is a vertex cover in $G$, for each $e_j = v_iv_k \in E$ we have at least one of $l_id_j$ or $l_kd_j$ is incident to a vertex with
price $1$. Form $F' \subseteq E'$ by adding one such edge from each $e_j$ and adding all node customers. We get profit $m$ from edge customers, profit $2(n-k)$ from all node customers $l_ir_i$ such that $v_i \notin S$, and profit $k$ from all node customers $l_ir_i$ with $v_i \in S$ for a total profit of $m + 2n - k$.

Conversely, consider an optimal pricing $p$, so each price is 1 or 2. For $e_j = v_iv_k \in E$, we claim that either $p(v_i) = 1$ or $p(v_k) = 1$.
If not, then consider changing $p(v_i)$ to 1. We lost a profit of 1 from the node customer $l_ir_i$ but have gained a profit of 1 by adding $v_id_j$, which remains feasible because neither $v_id_j$ nor $v_kd_j$ could afford the price of their left-endpoint before (i.e. $d_j$ is not used by any edge that can afford their price under pricing $p$, so we may add $v_id_j$ after adjusting prices).

Set $S = \{v_i : p(l_i) = 1\}$. By the above argument, $S$ is a vertex cover of $G$. Also observe that the optimal set of edges of $G'$
under prices $p$ will include every node customer plus exactly one from each pair $\{l_id_j, l_kd_j\}$ for each $e_j = v_iv_k \in E$.
So the profit of $p$ is $m + 2n - |S|$.

Therefore, the optimal profit in $G'$ is exactly $\frac{5}{2} \cdot n - k$ where $k$ is the size of a minimum vertex cover of $G$. There are constants $0 < \alpha < \beta < 1$ such that it is {\bf NP}-hard to distinguish between 3-regular graphs having vertex covers of size $\leq \alpha \cdot n$ and 3-regular graphs requiring vertex covers of size $\geq \beta \cdot n$. So it is {\bf NP}-hard to distinguish between \lsp instances that have optimal profit at least $\left(\frac{5}{2} - \alpha\right) \cdot n$ or at most $\left(\frac{5}{2} - \beta\right) \cdot n$.
\end{proof}


\section{Conclusion}\label{sec:conclusion}

We presented an 8-approximation for \cgp. If all capacities were bounded from above by a constant or, in simple graphs, were bounded from below by a sufficiently large constant then we get slightly better approximations. It would be nice to combine these two cases in a more clever way to beat the 8-approximation in any \cgp instance without an assumption on the uniformity of the vertex capacities, even if only for simple graphs. But the techniques we use are quite different and it is not clear how to combine them in a single algorithm that works in the presence of both small and large capacities.

It would also be interesting to know if the hardness lower bound for \cgp is worse than 4. Intuitively, this could be the case as the \lsp problem we reduce to is {\bf APX}-hard in the capacitated case.

We also briefly remark that it is simple to get an LP-based $O(k^2)$-approximation for the generalization of \cgp to hypergraphs where each hyperedge has size $\leq k$ by first reducing to a bipartite hypergraph where we are only allowed to use nonzero prices on one side $L$ and each hyperedge has only one vertex in $L$ (losing an $O(k)$ in the guarantee \cite{BB06}) and then using standard randomized rounding of the obvious generalization of our LP to this setting while losing only an additional $O(k)$ factor. This problem is a common generalization of the uncapacitated case which has a hardness of $\Omega(k^{1-\epsilon})$ \cite{CLN13}, and the $k$-\textsc{Set Packing} problem which has a hardness of $\Omega(k/\log k)$ \cite{HSS06}. One then wonders if \cgp in hypergraphs could be hard to approximate better than $\Omega(k^{2-\epsilon})$. It would be interesting to determine if this is the case or to see if there is a noticeably better approximation than $O(k^2)$, perhaps even $O(k)$.


\bibliographystyle{plain}
\bibliography{graphpricing}

\appendix


\section{Incorporating Loops} \label{app:loops}
The algorithms presented in this paper assumed every customer was interested in a bundle with precisely two distinct items. This was done for notational simplicity.
However, the algorithmic results extend very easily to the case where some customers may be only
interested in a single item. The reduction to \lsp is valid in this case as well and we only have to consider singleton customers interested in an item in $L$. One can still compute an optimum matching for a given pricing in this case, so the local search algorithm can still be executed. The analysis of the local search algorithms using test swaps can then be adapted in a straightforward way by removing singleton customers from the local optimum
and adding singleton customers from the global optimum who are interested in an item whose price changed when constructing the matching used to generate the inequality for this swap.

Similarly, the LP-based $(1+\epsilon)$-approximation for \lsp with large capacities from Section \ref{sec:lp} is is trivial to adapt. The ``edge-variables'' for singleton customers interested only an item in $L$ do not contribute
to the load of any constraint for any $v \in R$. The randomized rounding algorithm is identical.

\section{Efficient Versions of Local Search} \label{app:faster}
The standard trick to make local search algorithms efficient is to only make an improvement if it is somewhat noticeable. That is, a swap is performed only if it improves the cost
by a factor of at least $1+\eps/\Delta$ where $\Delta$ is the total ``weight'' of all inequalities generated by test swaps to complete the analysis (typically, $\Delta$ is polynomial in the input size).
See \cite{AGK04} for a specific example of this approach.

However, such analysis typically ``loses an $\epsilon$'' in the approximation guarantee.
We adapt an alternative approach outlined in \cite{FKS} that avoids this $\epsilon$-loss while still achieving the same approximation
guarantee that a true local optimum is proven to have. We consider the single-swap algorithm first, the extension to the multi-swap algorithm is in Section \ref{app:multi_fast}.

Recall that the proof of Theorem \ref{thm:localsingle} described a set of {\em test swaps} and placed a bound on the cost change.
That is, for each $u \in L$ the swap $p \rightarrow p^u$ is considered and a bound $\Delta_u$ on the change in $\texttt{val}()$
was given as
\[ \Delta_u \geq p^*(u) \cdot |\delta_{B^*}(u)| - p(u) \cdot |\delta_B(u)| - \sum_{e \in \delta_{B'}(u)} p(\sigma(e)). \]
Observe this bound holds even if $p$ is not a local optimum solution. The only place in the proof of Theorem \ref{thm:localsingle} that used the fact that $p$ was a local optimum was in asserting $0 \geq \Delta_u$, which is not required here.

Summing the above over all $u \in L$ shows
\[ \sum_{u \in L} \Delta_u \geq \texttt{val}(p^*) - 2 \cdot \texttt{val}(p). \]
Thus, the $u \in L$ with largest $\Delta_u$ satisfies
\[ \Delta_u \geq \frac{\texttt{val}(p^*) - 2 \cdot \texttt{val}(p)}{|L|}. \]
So if we take the best improvement in each step of the algorithm, the next price $p'$ then satisfies
\[ \texttt{val}(p') \geq \texttt{val}(p) + \frac{\texttt{val}(p^*) - 2 \cdot \texttt{val}(p)}{|L|}. \]

Consider the potential function $\Phi(p) := \texttt{val}(p^*) - 2 \cdot \texttt{val}(p)$. If $\Phi(p) > 0$, then $\Phi(p') \leq \left(1-\frac{2}{|L|}\right) \cdot \Phi(p)$ follows from the expression above. That is, $\Phi(p)$ decreases by a factor of $\exp(-1)$ after every $|L|/2$ iterations
as long as the current price $p$ satisfies $\Phi(p) > 0$.

With the standard assumption that the budgets $b_e$ are expressed as rational numbers in the input, after a polynomial number of iterations (in the total bit complexity of the input), we will reach a solution $p$ with $\Phi(p) \leq 0$, {\em i.e.} $\texttt{val}(p) \geq \texttt{val}(p^*)/2$
as required, provided we take the best improvement in each step.


\subsection{Extension to multi-swap}\label{app:multi_fast}
Each swap of the form $p \rightarrow p^{u,r}$ for $0 \leq i \leq d$ and $u \in L$ in the analysis was weighted with a value $0 \leq \tau(u,r) \leq C^{d-r}$. Let $\kappa = \frac{C^{d+1}-1}{C-1} \cdot |L|$, so $\kappa$
is an upper bound on the total weight of all test swaps and $\kappa = O(|L|)$ as $C$ and $d$ are constants.

Again, even if $p$ is not a local optimum our analysis still shows
\[
\sum_{u \in L, 0 \leq r \leq d} \tau(u, r) \cdot (\texttt{val}(p^{u,r}) - \texttt{val}(p)) \geq \frac{C^{d+1}-1}{C-1} \cdot \texttt{val}(p^*) - \left(\frac{C^{d+1}-1}{C-1} + C^d\right) \cdot \texttt{val}(p).
\]
Local optimality of $p$ was only used to show the left-hand side was not positive. Without local optimality, we may still conclude the most improving swap $p \rightarrow p'$ satisfies
\[
\texttt{val}(p') \geq \texttt{val}(p) + \frac{1}{\kappa} \cdot \left(\frac{C^{d+1}-1}{C-1} \cdot \texttt{val}(p^*) - \left(\frac{C^{d+1}-1}{C-1} + C^d\right) \cdot \texttt{val}(p)\right).
\]
Consider the potential function
\[
\phi(p) = \frac{C^{d+1}-1}{C-1} \cdot \texttt{val}(p^*) - \left(\frac{C^{d+1}-1}{C-1} + C^d\right) \cdot \texttt{val}(p).
\]
The above bound shows if $\phi(p) > 0$ then choosing the best improving swaps will result in a solution $p'$
with $\phi(p') \leq \left(1 - \left(\frac{C^{d+1}-1}{C-1} + C^d\right) \cdot \frac{1}{\kappa}\right) \cdot \phi(p)$.
So $\phi(p)$ decreases geometrically every $O(\kappa)$ iterations as long as it remains positive.
As $\kappa = O(|L|)$ and by using rationality of the input values, the potential $\phi(p)$ will become nonpositive after a polynomial number of iterations
in the total bit complexity of the input as long as we take the most improving swap.



\end{document}